\documentclass[10pt,a4paper]{article}
\usepackage{amsthm}
\usepackage{amsmath}
\usepackage{amsfonts}
\usepackage{amssymb}
\usepackage{graphicx}%
\usepackage[numbers]{natbib}
\newtheorem{theorem}{Theorem}

\newtheorem{example}[theorem]{Example}

\newtheorem{proposition}[theorem]{Proposition}
\newtheorem{remark}[theorem]{Remark}

\def\beas{\begin{eqnarray*}}
\def\eeas{\end{eqnarray*}}
\def\bea{\begin{eqnarray}}
\def\eea{\end{eqnarray}}
\def\be{\begin{equation}}
\def\ee{\end{equation}}
\def\bes{\begin{equation*}}
\def\ees{\end{equation*}}
\def\bi{\begin{itemize}}
\def\ei{\end{itemize}}

\renewenvironment{proof}[1][Proof]{\noindent\textbf{#1.} }{\ \rule{0.5em}{0.5em}}
\begin{document}
\title{Monte Carlo Greeks for financial products via approximative transition densities}
\author{J\"org Kampen$^1$, Anastasia Kolodko$^{1}$, and John Schoenmakers$^1$}
\maketitle \footnotetext[1]{Weierstrass Institute for Applied
Analysis and Stochastics, Mohrenstr. 39, D-10117 Berlin, Germany.
{\tt kampen[kolodko][schoenma]@wias-berlin.de}.
\\ Supported by the DFG Research Center \textsc{Matheon} `Mathematics for Key
Technologies' in Berlin.}

\begin{abstract}
In this paper we introduce efficient Monte Carlo estimators for the valuation of high-dimensional derivatives and their sensitivities (''Greeks''). These estimators are based on an analytical, usually approximative representation of the underlying density. We study approximative densities obtained by the WKB method. The results are applied in the context of a Libor market model.
\\[2mm] {\it Keywords:}   Financial derivatives,  sensitivities, Monte-Carlo methods, WKB expansions.
\\[0.6ex] {\it 2000 AMS subject
classification:} 60H10, 62G07, 65C05 
\end{abstract}

\section{Introduction}

Valuation methods for high-dimensional derivative products are typically based on Monte Carlo simulation of the underlying process. The dynamics of the underlyings are usually given via a (jump-)diffusion SDE. In case of a diffusion SDE, the underlying process may be simulated using an Euler scheme or a (weak) second order scheme {e.g. see \citet{KP} or \citet{MT0}. For simulation of jump-diffusions see e.g. \citet{Cont}, and \citet{GM} for simulation of (Libor) interest rate models with jumps.  

 The evaluation of option sensitivities, 'Greeks' in financial terms, comes down to the computation of expressions of the form $\frac{\partial }{\partial \lambda}E(f(X^{\lambda}))$ (and possibly higher order derivatives), where  $f$
is a pay-off function,   $X$ is the state of an underlying process depending on some parameter $\lambda.$ For example, the first and second order derivatives  with respect the initial state are called Deltas and Gammas, respectively. 
In the literature the evaluation of Greeks has been treated  by several methods (a nice overview about classical and recent literature is provided in \citet{EFT}). 
Classical finite difference approaches have been studied by \citet{EP},  \citet{BG}, \citet{MS}, \citet{MT}, \citet{DGR}, and \citet{GG06}. These approaches are quite general and easy to implement as they do not require particular knowledge 
of the distribution of the underlying. However, they require  full blown simulation of the corresponding system of stochastic differential equations and, in order to be efficient, some degree of regularity with respect to the pay-off function.
In case the transition kernel of $X$ is known or known in a good approximation, the latter drawback can be avoided by differentiating this kernel with respect to the sensitivity parameter $\lambda,$ see  \citet{FrKa} and \citet{FJ}.
The typical difficulty that the distribution of the underlying is only known for very special cases was overcome by
\citet{FLLLT}, who used the Malliavin integration-by-parts formula in order to express Greeks in the form
$E(f(X^{\lambda})\pi),$ where the random variable $\pi$ is called a {\em Greek weight.}  In a more recent alternative
approach 
\citet{EFT} construct Greek estimators which are based on variance minimizing choices of Greek weights. As a matter of fact,
the Malliavin method does not lead to this optimal weight in general.
In order to avoid straightforward SDE simulation in the context of the Libor market interest rate model, and so reducing simulation costs, \citet{KSS} considered lognormal approximations for the transition density,  whereas
\citet{HJJ}, and  \citet{PPR} propose specific drift approximations.

In an ideal situation, the density of the underlying process $X^\lambda$ at a fixed point in time is known explictly and an efficient method to sample from it is available.  Usually, however, neither of this is true. Even if the transition density is known, we will show that calculation of sensitivities, based on kernel differentiation for instance (as in \citep{FrKa}), may cause problems (high variance) in case the kernel under consideration is `highly peaked ', for example  due to small maturities, low volatilities, or high dimensionality of the underlying system. 
In this paper we therefore choose for a rather general approach with the following objectives.
\bi
\item
Developing  efficient variance bounded probabilistic representations  for price sensitivities, based on an
analytical approximation  of the underlying density  and a possibly rougher approximative standard density 
(e.g. a lognormal density)  which is basically used as an importance sampler. 
\item
 Construction  of  a ''good'' analytical approximation for the density of the underlying process by using (convergent) WKB\footnote{The historical origin of the name is the work of \textbf{W}entzel, \textbf{K}ramers, \textbf{B}rioullin in the context of semiclassical solutions of the Schr\"odinger equation. The meaning of WKB has broadened since; nowadays, it refers to analytic expansions of exponential form.} methods; 
\ei  
We underline that, in principle, the way of constructing an analytical approximation of the transition density is not essential for the developed Greek  estimators.  In this article  we exploit the use of WKB approximations as a generic convergent 
method.  In special cases, however, construction of high accuracy transition kernels may be possible by other means (see \citep{KSS} for example).

The structure of the paper is as follows.
In Section~2 we set up the model class for which we exemplify our methods and specify the financial products (including Bermudan callables) for  which prices and sensitivities  are to be determined. In Section~3 we introduce probabilistic representations for integral functionals of kernel type and their derivatives. As a particular result we prove that the corresponding estimator for the derivatives has non-exploding variance for sharply peaked kernels in contrast to some existing weighted Monte Carlo schemes. This estimator thus allows for 
efficient Monte Carlo estimation of option sensitivities, in particular with respect to underlyings (Deltas), even in situations where the densities are sharply peaked (for instance when volatilities are small). The general probabilistic representations introduced in Section~3 are applied to the computation of Deltas for Bermudan callable products in Section~4.
Section 5 deals with the WKB-theory of densities of diffusion equations (densities of processes which have continuous paths). 
%
%
In Section 5.1 we summarize some results concerning pointwise valid WKB-representations of densities obtained in \citet{Ka}. Since in practice only finitely many terms of a WKB expansion can be computed, it will be necessary to use a truncated form of the WKB-representation for actual computations. 
 In Section 5.2. we analyze the effect of this truncation error on approximations of solutions of Cauchy problems and their derivatives.
The case of non-autonomous diffusion models is discussed in Section 5.3. The results of Sections 2-5 are applied in Section~6 to the Libor market model. In Section 6.1. we compute explicitly the first three coefficients of the WKB representation of the Libor model density. In Section 6.2 we compute prices and Deltas in a case study of European swaptions.

\section{Basic setup}\label{Basic}

Let $X$ $=$ $(X^{1},...,X^{n})$ be a Markovian  process of financial derivative in $\mathbb{R}_{+}^{n}$ ($\mathbb{R}_{+}$ $:=$ $\{x:x>0\}$) 
under a given pricing
measure $P,$ connected with a given discounting numeraire $B,$ $B$ $>$ $0,$  on some filtered
probability space. For example, $X$ may represent a system of asset prices  or (Libor) interest
 rates.
A popular framework for  
the system $(X,B)$ is, for instance, the class of jump-diffusions (e.g. \citet{Cont}). For simplicity however,  we mainly 
consider in the present article ordinary diffusions,  
but, note that the main results generally extend to jump processes as well
(see \citet{KKS}).

With respect to an $n$-dimensional standard Wiener process $W$ $=$ $(W^{1},$ $...,$ $W^{n})^{\top}$ on the probability space $(\Omega,\mathcal{F},(\mathcal{F}_{t})_{t\in\lbrack t_{0},T]},P),$
where as usual $(\mathcal{F}_{t})$ is the $P$-augmentation of the filtration generated by $W,$
we assume that $X$ is governed by the stochastic differential equation (SDE),
\begin{equation} 
\ \frac{dX^{i}}{X^{i}}=\mu(t,X)dt+\sum_{j=1}^{n}\sigma^{ij}(t,X)dW^{j},\quad 1\le i,j\le n.\label{D1}%
\end{equation}
  It is assumed that $\mu(t,x)$ and the matrix
$\sigma(t,x)=\left(\sigma^{ij}(t,x)\right),$ $t\in\lbrack t_{0},T],$ $x\in
\mathbb{R}_{+}^{n}$ are such that for all
$x_{0}\in\mathbb{R}_{+}^{n},$ there exists a unique solution
$t\rightarrow X_{t}\in\mathbb{R}_{+}^{n}$ of
(\ref{D1}) for $t_{0}\leq t\leq T$ satisfying $X_{t_{0}}%
=x_{0}=:X_{t_{0}}^{t_{0},x_{0}}.$  
It is further assumed that the Markov process $X$ has a transition density 
\begin{equation}\label{Td}
p(t,x,s,y),\quad t_0\le s\le t\le T, \quad x,y\in \mathbb{R}_{+}^{n}, 
\end{equation} 
which is differentiable with respect to  $x,y,s,$ and $t,$  up to any order. In order to guarantee the existence and uniqueness of (\ref{D1}), and the existence of the transition density (\ref{Td}) as stated, it is sufficient to require  that the functions $\mu(\cdot,\cdot)$ and $\sigma(\cdot,\cdot)$ are bounded and have bounded derivatives up to any order, and that  the volatility matrix $\sigma(t,x)$
is regular with
\be\label{fuc}
0<\lambda_1\le\left|\left(\sigma\sigma^\top\right)(t,x)\right|\le\lambda_2
\ee
for all $(t,x),$ $t$ $\in$ $[t_0,T]$,  $x$ $\in$ $\mathbb{R}_{+}^{n}$, and some $0$ $<$ $\lambda_1$ $<$
$\lambda_2$ (see for example \citet{BT}). 

Let us take (w.l.o.g.) $B_{0}=1$ and consider a contingent claim with pay-off
function of the form $f(X_{\tau})B_{\tau}$ at some $(\mathcal{F}_{\cdot}%
)$-stopping time $\tau.$ By general arguments (e.g. \citet{Duf01}), the price of this claim at
time $t_{0}$ is given by%
\[
u(t_{0},x_{0})=E\ f(X_{\tau}^{t_{0},x_{0}}).
\]
For deterministic $\tau,$ say $\tau\equiv T,$ we have a European claim, and
for $t_{0}\leq t\leq T$ its discounted value process can be represented by%
\[
u_{t}:=u(t,X_{t}):=E^{{\mathcal F}_t} f(X_{T})=\int p(t,X_{t},T,y)f(y)dy,\quad\mbox{where}
\]
\begin{equation}
u(t,x)=\int p(t,x,T,y)f(y)dy \label{U}%
\end{equation}
is the unique solution of the Cauchy problem
\begin{align}
\frac{\partial u}{\partial t}+\frac{1}{2}\sum_{i,j=1}^{n}x^i x^j\left(\sigma\sigma^\top\right)^{ij}(t,x)\frac
{\partial^{2}u}{\partial x^{i}\partial x^{j}}+\sum_{i=1}^{n}x^{i}%
\mu(t,x)\frac{\partial u}{\partial x^{i}}  &  =0,\label{BK}\\
u(T,x)  &  =f(x). \nonumber
\end{align}
The density kernel $p(\cdot,\cdot,T,y)$ is the unique (weak) solution of
(\ref{BK}) with $p(T,x,T,y)$ $=$ $\delta(x-y),$ where $\delta$ is the
Dirac-delta function in Schwarz distribution sense.
\bigskip

Of particular importance are Bermudan callable contracts. A Bermudan contract starting at $t_{0},$ is specified by a set of exercise dates
$\{t_{1},t_{2},$ $...,t_{\mathcal{I}}\},$ where $t_{0}$ $<$ $t_{1}%
<...<t_{\mathcal{I}}$ $<T$, and corresponding (discounted) pay-off functions
$f_{i}(x),$ $1\leq i\leq$ $\mathcal{I}.$ According to the contract, the
holder has the right to call (once) a cash-flow $f_{i}(X_{t_{i}}^{t_{0},x_{0}%
})B_{t_{i}}^{t_{0},x_{0},1}$ (with $B_{0}^{t_{0},x_{0},1}=1$) at an exercise date $t_{i}$ of his choice. It is
well known (e.g. \citet{Duf01}) that the discounted price of this contract at time $t,$ $t_{0}\leq t\leq
T,$ assuming that no exercise took place before $t,$ is given by
\begin{equation}
u(t,x):=\sup_{\tau\in\mathcal{T}_{i,\mathcal{I}}}Ef_\tau(X_{\tau}%
^{t,x})=Ef_{\tau_{\ast}^{t,x}}(X_{\tau_{\ast}^{t,x}}^{t,x}),\quad t_{i-1}<t\leq t_{i}%
,\label{eqref}%
\end{equation}
where $x=X_{t}^{t_{0},x_{0}},$  $\mathcal{T}%
_{i,\mathcal{I}}$ the set of stopping times $\tau$ taking values in
$\{t_{i},t_{i+1},...,t_{\mathcal{I}}\},$ and $\tau_{\ast}^{t,x}$ is an optimal
stopping time. In particular the process $u(t,X_t)$ is a supermartingale and is called the Snell envelope
of the (discounted) cash-flow process $f_i(X_{t_i})$.  

\section{Probabilistic representations and their estimators
}\label{Estimators}
In this section we consider for a given smooth function $u:\mathbb{R}_{+}^{n}\rightarrow\mathbb{R}_{+}$
and a smooth kernel function $p:\mathbb{R}_{+}^{n}\times\mathbb{R}_{+}%
^{n}\rightarrow\mathbb{R}_{+},$  probabilistic representations for the integral
\begin{eqnarray}\notag
I(x)&:=&\int p(x,y)u(y)dy,\label{I(x)}\quad \mbox{and its gradient}\\
\frac{\partial I}{\partial x}(x)&=&\int \frac{\partial}{\partial x}p(x,y)u(y)dy, \quad   \label{pu}
\mbox{with}\quad
\frac{\partial}{\partial x}
:=
 \left(\frac{\partial}{\partial x_1},\ldots,\frac{\partial}{\partial x_n}\right).
\end{eqnarray} 
Here and in the following  sufficient (uniform) integrability conditions are assumed to be fulfilled, for instance, in order to guarantee that (\ref{pu}) is valid.
\begin{remark} 
{\rm
In (\ref{pu}), kernel $p$ (which may or may not be a density in the second argument) and function $u$ have to be distinguished from  the respective definitions in Section~\ref{Basic}, 
although they may be related. For fixed $t,T,$ $0\le t\le T,$ one could take (see (\ref{U}-\ref{BK})), $p(x,y)$ $:=$ $p(t,x,T,y)$ and $u(x)$ $:=$ $u(t,x)$ for example. 
}
\end{remark}
Let $\zeta$ be an $\mathbb{R}_{+}%
^{n}$-valued  random variable on some probability space with density $\phi,$  $\phi>0.$ Then, obviously,%
\begin{equation}
I(x)=E\ p(x,\zeta)\frac{u(\zeta)}{\phi(\zeta)}\label{Pr1}%
\end{equation}
is a probabilistic representation for (\ref{I(x)}) which may be estimated by
the unbiased Monte Carlo estimator%
\begin{equation}
\widehat{I}(x):=\frac{1}{M}\sum_{m=1}^{M}p(x,_{m}\zeta)\frac{u(_{m}\zeta
)}{\phi(_{m}\zeta)},\label{E1}%
\end{equation}
where for $m=1,...,M,$ $_{m}\zeta$ are i.i.d. samples from a distribution with
density $\phi.$ By taking gradients in (\ref{Pr1}) we readily obtain the
probabilistic representation%
\begin{equation}
\frac{\partial I}{\partial x}(x)=E\ \frac{\partial}{\partial x}p(x,\zeta
)\frac{u(\zeta)}{\phi(\zeta)},\label{Prdel}%
\end{equation}
with corresponding estimator,%
\begin{equation}
\widehat{\frac{\partial I}{\partial x}}(x):=\frac{1}{M}\sum_{m=1}^{M}%
\frac{\partial}{\partial x}p(x,_{m}\zeta)\frac{u(_{m}\zeta)}{\phi(_{m}\zeta
)}.\label{Naiv}%
\end{equation}
While as a rule (\ref{E1}) is an effective estimator for $I(x)$ for a proper choice of $\phi$, unfortunately the
gradient estimator (\ref{Naiv}) has a serious drawback: If the kernel
$p(x,\cdot)$ is sharply peaked (nearly proportional to a 'delta-function'), its variance may be extremely
high. This fact is demonstrated by the following stylistic example of a
multi-asset model, which is nevertheless realistic  in orders of magnitude.

\begin{example}\label{ExExpl}{\rm
Consider for fixed $x_{0}\in\mathbb{R}_{+}^{n},$ parameters $s>0,$ and
$\sigma>0,$ the $n$-dimensional lognormal density
\begin{equation}
p(s,\sigma;x_{0},y):=\frac{1}{\left(  2\pi\sigma^{2}s\right)  ^{n/2}}%
{\displaystyle\prod\limits_{i=1}^{n}}
\frac{\exp\left[  -\frac{1}{2\sigma^{2}s}\ln^{2}\frac{y^{i}}{x_{0}^{i}%
}\right]  }{y^{i}}.\label{LND}%
\end{equation}
In (\ref{LND}) $p(s,\sigma;x_{0},\cdot)$ is the density of the random
variable
$
(x_{0}^{1}e^{\sigma\sqrt{s}\xi^{1}},...,x_{0}^{n}e^{\sigma\sqrt{s}\xi^{n}}),
$ 
where $\xi^{i\text{ }},$ $i=1,...,d,$ are i.i.d. standard normal random
variables. Thus, for small $s$ and $\sigma$, $p(s,\sigma;x_{0},\cdot)$ is peaked
('delta-shaped') around $x_{0}.~$Let us now take $\phi(\cdot):=$%
\ $p(s,\sigma;x_{0},\cdot)$ in (\ref{Pr1}) and (\ref{Prdel}), respectively,
and $u\equiv||x_{0}||$ (a constant of order $x_{0}$ in magnitude). Clearly,
estimator (\ref{E1}) equals $||x_{0}||$ almost surely and so has zero
variance. However, estimator (\ref{Naiv}) is not deterministic and we have%
\begin{align*}
\widehat{\frac{\partial I}{\partial x^{j}}}(x_{0}) &  :=
\frac{1}{M}\sum_{m=1}^{M}\frac{||x_{0}||}{p(s,\sigma;x_{0},_{m}\zeta)}
\frac{\partial}{\partial x^{j}}p(s,\sigma;x_{0},_{m}\zeta)\\
&
=\frac{||x_{0}||}{M}\sum_{m=1}^{M}\frac{\partial}{\partial x^{j}}\ln p(s,\sigma;x_{0},_{m}\zeta)\\
&  
=\frac{||x_{0}||}{M}\sum_{m=1}^{M}\frac{\ln\frac{_{m}\zeta^{j}}{x_{0}^{j}}%
}{\sigma^{2}sx_{0}^{j}}=\frac{||x_{0}||}{M}\sum_{m=1}^{M}\frac{_{m}\xi^{1}%
}{\sigma\sqrt{s}x_{0}^{j}}.
\end{align*}
Hence, $E\left[  \widehat{\frac{\partial I}{\partial x^{j}}}(x_{0})\right]
=0$ as should be, but,
\begin{equation} \label{explo}
\text{Var}\left[  \widehat{\frac{\partial I}{\partial x^{j}}}(x_{0})\right]
=\frac{||x_{0}/x_{0}^{j}||^{2}}{M}\frac{1}{\sigma^{2}s}%
\end{equation} 
which explodes when $\sigma^{2}s$ goes to zero!
}
\end{example}

\begin{remark}{\rm
In \citet{FrKa} estimators (\ref{E1}) and (\ref{Naiv}) are used for computing prices and sensitivities of European Libor options, respectively. In their numerical examples they used $50\%$ (rather high) volatility in order to amplify Monte Carlo errors. While, indeed,  a larger volatility generally gives rise to a large Monte Carlo error of (\ref{E1}), Example~\ref{ExExpl} shows that the opposite is true for estimator (\ref{Naiv}).  For example, $50\%$ volatility  in combination with $0.5\,$yr. maturity corresponds to a (just moderate) variance factor $1/\left(\sigma^{2}s\right)$ $=$ $8.0$ in (\ref{explo}), while a more usual Libor volatility, e.g. $14\%$, and $0.5\,$y maturity would give a factor $102.0$(!). 
}
\end{remark}

In the present paper we propose sensitivity estimators which are efficient on a broad time and volatility scale. As a result, the next theorem provides a tool for constructing sensitivity (gradient) estimators with non-exploding variance.   

\begin{theorem}
\label{The} Let $\lambda$ be a reference density on $\mathbb{R}^{n}$ with
$\lambda(z)\neq0$ for all $z$ (for example, the standard normal density). Let
$\xi$ be an $\mathbb{R}^{n}$-valued random variable on some probability space, with density $\lambda$ and
$g:\mathbb{R}_{+}^{n}\times\mathbb{R}^{n}\rightarrow\mathbb{R}_{+}^{n}$ be a
smooth enough map which has at least continuous derivatives with $\left\vert
\partial g(x,z)/\partial z\right\vert \neq0,$ such that for each
$x\in\mathbb{R}_{+}^{n}$ the random variable $\zeta^{x}:=g(x,\xi)$ has a
density $\phi(x,\cdot)$ on $\mathbb{R}_{+}^{n}.$ Then, for (\ref{pu}) we have the
probabilistic representation 
\begin{equation}
\frac{\partial I}{\partial x}(x)=E\,\frac{\partial}{\partial x}\frac
{p(x,\zeta^{x})u(\zeta^{x})}{\phi(x,\zeta^{x})}=E\,\frac{\partial}{\partial
x}\frac{p(x,g(x,\xi))u(g(x,\xi))}{\phi(x,g(x,\xi))},\label{Istr}%
\end{equation}
with corresponding Monte Carlo estimator%
\begin{equation}
\widehat{\frac{\partial I}{\partial x}}(x)=\frac{1}{M}\sum_{m=1}^{M}%
\frac{\partial}{\partial x}\frac{p(x,g(x,_{m}\xi))u(g(x,_{m}\xi))}%
{\phi(x,g(x,_{m}\xi))}.\label{nonexpl}%
\end{equation}
Let $\left\Vert \cdot\right\Vert _{\alpha}:=\sqrt[\alpha]{E\left\vert
\cdot\right\vert ^{\alpha}}$ where $\left\vert \cdot\right\vert $ denotes
either a vector norm or a compatible matrix norm. Then it holds
\begin{equation}
E\,\left\vert \frac{\partial}{\partial x}\frac{p(x,g(x,\xi))u(g(x,\xi))}%
{\phi(x,g(x,\xi))}\right\vert ^{2}\leq2M_{2}^{2}M_{3}^{2}M_{4}^{2}+4M_{1}%
^{2}M_{4}^{2}M_{5}^{2}+4M_{1}^{2}M_{3}^{2}M_{4}^{2}M_{6}^{2},\label{Varest}%
\end{equation}
hence the second moments of the Monte Carlo samplers for the components of
$\partial I/\partial x$ are bounded by the right-hand-side of (\ref{Varest}),
if for fixed $x\in\mathbb{R}_{+}^{n},$ there are constants $\alpha
_{1},...,\alpha_{6}>1$ and $M_{1},...,M_{6}>$ $0$ with
\[
\frac{1}{\alpha_{4}}+\frac{1}{\alpha_{1}}+\frac{1}{\alpha_{5}}=1,\quad\frac
{1}{\alpha_{4}}+\frac{1}{\alpha_{2}}+\frac{1}{\alpha_{3}}=1,\quad\frac
{1}{\alpha_{4}}+\frac{1}{\alpha_{1}}+\frac{1}{\alpha_{6}}+\frac{1}{\alpha_{3}%
}=1,
\]
such that,%
\begin{align}
\left\Vert u(g(x,\xi))\right\Vert _{2\alpha_{1}}  & \leq M_{1},\quad\left\Vert
\frac{\partial u}{\partial y}(g(x,\xi))\right\Vert _{2\alpha_{2}}\leq M_{2},\notag\\
\left\Vert \frac{\partial g}{\partial x}(x,\xi)\right\Vert _{2\alpha_{3}}  &
\leq M_{3},\quad\left\Vert \frac{p(x,g(x,\xi))}{\phi(x,g(x,\xi))}\right\Vert
_{2\alpha_{4}}\leq M_{4}, \label{Ms}
\end{align}%
\[
\left\Vert \left(  \frac{1}{p}\frac{\partial p}{\partial x}-\frac{1}{\phi
}\frac{\partial\phi}{\partial x}\right)  (x,g(x,\xi))\right\Vert _{2\alpha
_{5}}\!\!\leq M_{5},\ \left\Vert \left(  \frac{1}{p}\frac{\partial p}{\partial
y}-\frac{1}{\phi}\frac{\partial\phi}{\partial y}\right)  (x,g(x,\xi
))\right\Vert _{2\alpha_{6}}\!\!\leq M_{6}.
\]

\end{theorem}

\begin{proof}
For any bounded measurable $\psi:\mathbb{R}_{+}^{n}\rightarrow\mathbb{R},$ we
have%
\[
\int\psi(g(x,z))\lambda(z)dz=\int\psi(g(x,z))\phi(x,g(x,z))\left\vert
\frac{\partial g(x,z)}{\partial z}\right\vert dz.
\]
Therefore, the densities $\phi$ and $g$ are connected via the relationship%
\begin{equation}
\phi(x,g(x,z))\left\vert \frac{\partial g(x,z)}{\partial z}\right\vert
=\lambda(z). \label{meas}%
\end{equation}
By (\ref{meas}), the right-hand-side of
(\ref{Istr}) equals 
\begin{align*}
\frac{\partial}{\partial x}E\,\frac{p(x,g(x,\xi))u(g(x,\xi))}{\phi
(x,g(x,\xi))}  &  =\frac{\partial}{\partial x}\int\frac{p(x,g(x,z))u(g(x,z))}%
{\phi(x,g(x,z))}\lambda(z)dz\\
&  =\frac{\partial}{\partial x}\int p(x,g(x,z))u(g(x,z))\left\vert
\frac{\partial g(x,z)}{\partial z}\right\vert dz\\
&  =\frac{\partial}{\partial x}\int p(x,y)u(y)dy=\frac{\partial I}{\partial
x}(x).
\end{align*}
To prove the moment estimation (\ref{Varest}), we observe that%
\begin{align*}
&  E\,\left\vert \frac{\partial}{\partial x}\frac{p(x,\zeta^{x})u(\zeta^{x}%
)}{\phi(x,\zeta^{x})}\right\vert ^{2}=E\,\left\vert \frac{\partial}{\partial
x}\frac{p(x,g(x,\xi))u(g(x,\xi))}{\phi(x,g(x,\xi))}\right\vert ^{2}\\
&  =E\,\left\vert u(g(x,\xi))\frac{\partial}{\partial x}\frac{p(x,g(x,\xi
))}{\phi(x,g(x,\xi))}+\frac{p(x,g(x,\xi))}{\phi(x,g(x,\xi))}\frac{\partial
u}{\partial y}(g(x,\xi))\frac{\partial g}{\partial x}(x,\xi)\right\vert ^{2}\\
&  \leq2E\,\frac{p^{2}(x,g(x,\xi))}{\phi^{2}(x,g(x,\xi))}\left\vert
\frac{\partial u}{\partial y}(g(x,\xi))\right\vert ^{2}\left\vert
\frac{\partial g}{\partial x}(x,\xi)\right\vert ^{2}\\
&  \quad+2E\,u^{2}(g(x,\xi))\left\vert \frac{\partial}{\partial x}%
\frac{p(x,g(x,\xi))}{\phi(x,g(x,\xi))}\right\vert ^{2}
=:2(I)+2(II).
\end{align*}
Then by H\"{o}lders inequality, $(I)$ $\leq$
\[
\sqrt[\alpha_{2}]{E\,\left\vert \frac{\partial u}{\partial y}(g(x,\xi
))\right\vert ^{2\alpha_{2}}}\sqrt[\alpha_{3}]{E\,\left\vert \frac{\partial
g}{\partial x}(x,\xi)\right\vert ^{2\alpha_{3}}}\sqrt[\alpha_{4}]%
{E\frac{p^{2\alpha_{4}}(x,g(x,\xi))}{\phi^{2\alpha_{4}}(x,g(x,\xi))}}\leq M_{2}^{2}M_{3}^{2}M_{4}^{2}.
\]
For the second term we have $(II)$ $=$
\begin{align*}
&  E\,u^{2}(g(x,\xi))\,\frac{p^{2}(x,g(x,\xi))}{\phi^{2}(x,g(x,\xi
))}\left\vert \frac{\frac{\partial p}{\partial x}(x,g(x,\xi))}{p(x,g(x,\xi
))}-\frac{\frac{\partial\phi}{\partial x}(x,g(x,\xi))}{\phi(x,g(x,\xi
))}\right. \\
&  \left.  +\left(  \frac{\frac{\partial p}{\partial y}(x,g(x,\xi
))}{p(x,g(x,\xi))}-\frac{\frac{\partial\phi}{\partial y}(x,g(x,\xi))}%
{\phi(x,g(x,\xi))}\right)  \frac{\partial g}{\partial x}(x,\xi)\right\vert
^{2}\\
&  \leq2E\,u^{2}(g(x,\xi))\,\frac{p^{2}(x,g(x,\xi))}{\phi^{2}(x,g(x,\xi
))}\left\vert \frac{\frac{\partial p}{\partial x}(x,g(x,\xi))}{p(x,g(x,\xi
))}-\frac{\frac{\partial\phi}{\partial x}(x,g(x,\xi))}{\phi(x,g(x,\xi
))}\right\vert ^{2}\\
&  +2E\,u^{2}(g(x,\xi))\,\frac{p^{2}(x,g(x,\xi))}{\phi^{2}(x,g(x,\xi
))}\left\vert \frac{\frac{\partial p}{\partial y}(x,g(x,\xi))}{p(x,g(x,\xi
))}-\frac{\frac{\partial\phi}{\partial y}(x,g(x,\xi))}{\phi(x,g(x,\xi
))}\right\vert ^{2}\left\vert \frac{\partial g}{\partial x}(x,\xi)\right\vert
^{2}\\
&  \leq2M_{1}^{2}M_{4}^{2}M_{5}^{2}+2M_{1}^{2}M_{3}^{2}M_{4}^{2}M_{6}^{2},
\end{align*}
again by H\"{o}lders inequality.
\end{proof}

\begin{remark}{\rm
If in (\ref{Ms}) the random variables $u(g(x,\xi)),$ $\frac{\partial u}{\partial y}(g(x,\xi)),$ and so on, have moments of high enough order, Theorem~\ref{The} guarantees that the variance of estimator (\ref{nonexpl}) is controlled via the moment estimates (\ref{Ms}).
The most delicate bound in (\ref{Ms})  is $M_5$ in fact. Indeed, if one takes $g(x,\xi)\equiv g(x_0,\xi)$ estimator 
(\ref{nonexpl}) collapses to (\ref{Naiv}), and   in Example~\ref{ExExpl}, page \pageref{ExExpl}, where $\phi(x,y)\equiv p(x_0,y)$ in fact, we see that $M_5$ cannot be taken small when $\sigma^2s$ is small, i.e. when $p$ is highly peaked around $x$. In contrast, if for fixed $x,$  $\phi(x_,\cdot)$ is approximately proportional to $p(x,\cdot)$  and $\partial\ln \phi(x,\cdot)/\partial x$ $\approx$  $\partial\ln p(x,\cdot)/\partial x$ (both with respect to the weight function $\phi(x,\cdot)$), a small $M_5$ may exist. Note that for  $\phi(\cdot,\cdot)$ exactly proportional to $p(\cdot,\cdot),$ we may take $M_5=0.$ 
}
\end{remark}

\begin{remark}{\rm
It can be shown that  Theorem~\ref{The} can be extended to probabilistic representations and corresponding estimators for higher order derivatives,
\begin{equation*}
\frac{\partial I}{\partial x^\beta}(x)=E\,\frac{\partial}{\partial x^\beta}\frac
{p(x,\zeta^{x})u(\zeta^{x})}{\phi(x,\zeta^{x})}=E\,\frac{\partial}{\partial
x^\beta}\frac{p(x,g(x,\xi))u(g(x,\xi))}{\phi(x,g(x,\xi))},
\end{equation*}
with corresponding Monte Carlo estimator%
\begin{equation}
\widehat{\frac{\partial I}{\partial x^\beta}}(x)=\frac{1}{M}\sum_{m=1}^{M}%
\frac{\partial}{\partial x^\beta}\frac{p(x,g(x,_{m}\xi))u(g(x,_{m}\xi))}%
{\phi(x,g(x,_{m}\xi))},\label{nonexpl_k}
\end{equation}  
where $\beta:=(\beta_1,\ldots,\beta_n)$, $\beta_i$  $\in$ $\{0,1,2,\ldots\}$ is a multi-index with (formally) $\partial x^\beta$ $=$ $\partial x_1^{\beta_1}\partial x_2^{\beta_2}\cdots\partial x_n^{\beta_n}.$
Loosely speaking, the variance of the higher order derivative estimator (\ref{nonexpl_k}) can 
be bounded from above by an expression like (\ref{Varest}) involving (i)
sufficiently high moments of the derivatives, 
$
y\rightarrow\frac{\partial u}{\partial y^\gamma},\quad \mbox{and}\quad z\rightarrow\frac{\partial g(x,z)}{\partial z^\gamma}
$ 
for fixed $x$, $\gamma$
$\le$ $\beta$ (component wise), with respect to weight functions $y\rightarrow\phi(x,y)$ and $z\rightarrow\lambda(z),$ respectively,
and, (ii) for fixed $x,$ $L^q(\mathbb{R}^n_+,\phi(x,y)dy)$-norms  of
\bes
y\rightarrow\frac{\partial }{\partial x^\gamma}\left(\frac{\phi(x,y)}{p(x,y)}\right),\quad \mbox{and}\quad
y\rightarrow\frac{\partial }{\partial y^\gamma}\left(\frac{\phi(x,y)}{p(x,y)}\right),\quad \gamma\le\beta,
\ees
for $q$ large enough.
}  
\end{remark}

\begin{remark}{\rm
In the next section we consider financial applications where $I(x)$ is the price of a derivative contract considered in dependence of the argument $x$ which may stand for the underlying process or some parameter (vector) which affects the dynamics of the underlying process (e.g. volatilities). Moreover, we there give a recipe how to construct a lognormal density approximation $\phi,$ corresponding to a particular normal reference density  and an exponential type  transformation  (in Theorem~\ref{The},  $\lambda$ and $g$ respectively). 
}
\end{remark}

\begin{remark}{\rm
Theorem~4 can be can be generalized to the case where both $p$ and $\phi$ live on a common submanifold rather than on the whole state space. Via such a generalization it would be possible to extend our results to factor reduced situations in the spirit 
of \citet{FrKa} and \citet{FJ}. However, this is considered beyond the scope of the present article and therefore we restrict ourselves to the full-factor case.  
}
\end{remark}

\section{Sensitivities for Bermudan options}\label{Bermudan}

Theorem~\ref{The} may  be applied in general for computing sensitivities (''Greeks'') of derivative products.
For estimator (\ref{Naiv}) the danger of exploding variance is typically the largest when derivatives of prices with respect to underlyings (Deltas, Gammas) are considered. We therefore consider  in this section
only (first order) derivatives with respect to the underlying process, hence Deltas.    

Let $\tau:\Omega\rightarrow\mathbb{R}_{+}$ be a given stopping time with
respect to the filtration $(\mathcal{F}_{\cdot}).$ As usual we may think of $\Omega$ as being the space of 
functions $\omega$ $:$ $[0,\infty)$ $\rightarrow$ $\mathbb{R}^{n},$ which are continuous from the right and have limits  from the left, and 
define $\tau
^{s,x}(\omega)$ $:=$ $s+\tau(X_{s+(\cdot)}^{s,x}(\omega)).$
We  now consider the Bermudan contract introduced in Section~\ref{Basic}.
For fixed $t,t^{+},$ $t_{0}\leq t\leq t^{+}\leq t_{1},$ $x\in\mathbb{R}%
_{+}^{n},$ we have $\tau_{\ast}^{t,x}=\tau_{\ast}^{t^{+},X_{t^{+}}^{t,x}}$
since $\tau_{\ast}^{t,x}\geq t_{1},$ and we thus may write%
\begin{align*}
u(t,x) &  :=Ef(X_{\tau_{\ast}^{t,x}}^{t,x})=EE^{\mathcal{F}_{t^{+}}}%
f(X_{\tau_{\ast}^{t^{+},X_{t^{+}}^{t,x}}}^{t^{+},X_{t^{+}}^{t,x}})\\
&  =\int p(t,x,t^{+},y)dyEf(X_{\tau_{\ast}^{t^{+},z}}^{t^{+},y})
=\int p(t,x,t^{+},y)u(t^{+},y)dy,
\end{align*}
by the Chapman-Kolmogorov equation. 

For each $t,t^+$ as above, let $\phi(t,x,t^{+},y),$ $g(t,x,t^{+},y)$, and reference density $\lambda(t,t^{+},z)$ be as in Theorem~\ref{The}. We then have the probabilistic representation
\begin{equation}
u(t,x)=E\ \frac{p(t,x,t^{+},g(t,x,t^{+},\xi))}{\phi(t,x,t^{+},g(t,x,t^{+}%
,\xi))}f(X_{\tau_{\ast}^{t^{+},g(t,x,t^{+},\xi)}}^{t^{+},g(t,x,t^{+},\xi)}),\label{pr}
\end{equation}
with Monte Carlo estimator
\begin{equation}
\widehat u(t,x):=\frac{1}{M}\sum_{m=1}^{M}  \frac{p(t,x,t^{+},g(t,x,t^{+},_{m}\xi))}{\phi(t,x,t^{+}%
,g(t,x,t^{+},_{m}\xi))}f(X_{\tau_{\ast}^{t^{+},g(t,x,t^{+},\ _{m}\xi)}}%
^{t^{+},g(t,x,t^{+},\ _{m}\xi)}),\label{E11}
\end{equation}
and for the gradients (Deltas) we have the probabilistic representation%
\begin{equation}
\Delta_{i}:=\frac{\partial u}{\partial x^{i}}(t,x)=E\ \frac{\partial}{\partial
x^{i}}\left(  \frac{p(t,x,t^{+},g(t,x,t^{+},\xi))}{\phi(t,x,t^{+}%
,g(t,x,t^{+},\xi))}f(X_{\tau_{\ast}^{t^{+},g(t,x,t^{+},\xi)}}^{t^{+}%
,g(t,x,t^{+},\xi)})\right)  \label{del}%
\end{equation}
with Monte Carlo estimator
\begin{equation}
\widehat{\Delta}_{i}:=\frac{1}{M}\sum_{m=1}^{M}\frac{\partial}{\partial x^{i}%
}\left(  \frac{p(t,x,t^{+},g(t,x,t^{+},_{m}\xi))}{\phi(t,x,t^{+}%
,g(t,x,t^{+},_{m}\xi))}f(X_{\tau_{\ast}^{t^{+},g(t,x,t^{+},\ _{m}\xi)}}%
^{t^{+},g(t,x,t^{+},\ _{m}\xi)})\right)  ,\label{delprop1}%
\end{equation}
where $_{m}\xi,$ $m=1,...,M,$ are i.i.d. samples from the reference density
$\lambda.$ Indeed, by pre-conditioning on $\mathcal{F}_{t^{+}}$ and then taking expectations we see that   
(\ref{E11}) and (\ref{delprop1}) are unbiased Monte Carlo estimators for the price (\ref{pr}) and 'deltas' (\ref{del}),
respectively. 
Moreover, if $\phi$ is close to $p$ in the sense of Theorem~\ref{The}, it is not difficult to see that also gradient estimator (\ref{delprop1}) has 
non-exploding variance when $t^+\downarrow t.$

Estimators (\ref{E11})  and (\ref{delprop1}) are  useful if one has an analytic approximation
$\widehat{p}(t,x,t^+,y)$ of the density $p(t,x,t^+,y)$  and known  densities $\phi(x,\cdot)$ for $x\in\mathbb{R}^n_+$. 
The approximation $\widehat{p}$ may be
obtained by
some specific method, for example by a WKB expansion as presented in Section~\ref{WKBsec}, or
some lognormal approximation as proposed in \citet{KSS} for the Libor market model.
Of course the density
$\phi$ has to be chosen with some care. If it is possible to sample directly from
$\widehat{p}$ (e.g. in case of a log-normal approximation) we may take $\phi$ $=$
$\widehat{p}.$ If not, (e.g. in the case of a WKB expansion) one may take for
$\phi$ a (not necessarily very accurate) lognormal approximation of the
density $p.$ 

A canonical lognormal approximation for $p(t,x,t^{+},z)$ is obtained by
freezing $X$ in the coefficients of (\ref{D1}) at the initial time. We thus
obtain
\begin{align}\label{zeta}%
X_{t^{+}}^{^{\mbox{\tiny lgn}}t,x;i}  & :=x^{i}\exp\left(  -\frac{1}{2}\sum_{j=1}^{n}%
\int_{t}^{t^{+}}(\sigma^{ij})^{2}(s,x)ds+\int_{t}^{t^{+}}r(s,x)ds\right.
\nonumber\\
& \left.  +\sum_{j=1}^{n}\int_{t}^{t^{+}}\sigma^{ij}(s,x)dW_{s}^{j}\right)
=:x_i\exp(\xi_i).
\end{align}
Here, $(\xi_i)_{i=1}^n$ is a Gaussian random vector with
$$
E\xi_i=-\frac{1}{2}\sum_{j=1}^{n}%
\int_{t}^{t^{+}}(\sigma^{ij})^{2}(s,x)ds+\int_{t}^{t^{+}}r(s,x)ds=:\mu^{i;t,t^+,x},\quad 1\leq i\leq n,
$$
and
$$
Cov(\xi_i,\xi_j)=\sum_{l=1}^n\int_t^{t^+}\sigma^{il}(s,x)\sigma^{jl}(s,x)\,ds=:\sigma^{ij;t,t^+,x},
\quad 1\leq i,j\leq n.
$$
Clearly, the density $\phi$ is then
given by
\begin{eqnarray}\label{LNf}
\phi(t,x,t^+,y):=\frac{\psi_{\mu^{t,t^{+},x},\sigma^{t,t^{+},x}}
(\ln \frac{y^{1}}{x^{1}},\ln \frac{y^{2}}{x^{2}},...,\ln \frac{y^{n}}{x^{n}%
})}{y^{1}y^{2}\cdot\cdot\cdot\,y^{n}}, 
\end{eqnarray}
$y^{i}>0,$ $1\leq i\leq n,$
with $\psi_{\mu^{t,t^{+},x},\sigma^{t,t^{+},x}}$ being the
density of the $n$-dimensional normal distribution $\mathcal{N}_{n}%
(\mu^{t,t^{+},x},\sigma^{t,t^{+},x})$ with $\mu^{t,t^{+}%
,x}:=(\mu^{i;t,t^{+},x})_{1\leq i\leq n}$ and 
$\sigma^{t,t^+,x}:=(\sigma^{ij;t,t^+,x})_{1\leq i,j\leq n}.$

For practical applications it is useful to discretize estimator (\ref{delprop1}) to  
\begin{eqnarray}
\widehat{\Delta}^h_{i}:=\frac{1}{M}\sum_{m=1}^{M}\frac{1}{2h%
}\left(  \frac{p(t,x+\mathfrak{h}_i,t^{+},g(t,x+\mathfrak{h}_i,t^{+},_{m}\xi))}{\phi(t,x+\mathfrak{h}_i,t^{+}%
,g(t,x+\mathfrak{h}_i,t^{+},_{m}\xi))}f(X_{\tau_{\ast}^{t^{+},g(t,x+\mathfrak{h}_i,t^{+},\ _{m}\xi)}}%
^{t^{+},g(t,x+\mathfrak{h}_i,t^{+},\ _{m}\xi)})\right.\notag\\
-\left.  \frac{p(t,x-\mathfrak{h}_i,t^{+},g(t,x-\mathfrak{h}_i,t^{+},_{m}\xi))}{\phi(t,x-\mathfrak{h}_i,t^{+}%
,g(t,x-\mathfrak{h}_i,t^{+},_{m}\xi))}f(X_{\tau_{\ast}^{t^{+},g(t,x-\mathfrak{h}_i,t^{+},\ _{m}\xi)}}%
^{t^{+},g(t,x-\mathfrak{h}_i,t^{+},\ _{m}\xi)})
\right),\label{delprop2}%
\end{eqnarray}
where $\mathfrak{h}_i:=h(\delta_{i1},\ldots,\delta_{in})$ ($\delta_{ij}$ being the Kronecker symbol), for small enough $h>0.$ 
Without further details we note that according to Milstein and Tretyakov (2004) in a related context, it is efficient to take $h\approx x/\sqrt{M}.$  

As an alternative, it is also possible to expand the derivatives in (\ref{delprop1}), which leads to a SDE system of first order variation as in \cite{MS} and \cite{GG06}. In the differentiation of (\ref{delprop1}) with respect to $x$ for a fixed trajectory,  $\tau_{\ast}^{t^{+},x}(\omega)$ can be considered to be independent  of $x$.
This can be seen as follows: if $\tau_{\ast}^{t^{+},x}(\omega)$ $=$ $p,$ the random variable $X^{t^{+},x}_p,$ which is assumed to have a  density in $\mathbb{R}^n,$  lays almost surely in the interior of the exercise region. Due to the fact that (almost surely)
the map $x$ $\rightarrow$ $X^{t^{+},x}_{p}$ is smooth (e.g. see Protter (1990)), $X^{t^{+},y}_{p}$ lays in the exercise region for $y$
in an open disc around $x.$ As a consequence, for any $y$ in this disc we have  $\tau_{\ast}^{t^{+},y}(\omega)$ $=$ 
$\tau_{\ast}^{t^{+},x}(\omega)$ $=$ $p.$
Thus, by
differentiating (\ref{delprop1}) path-wise we obtain
\begin{align}
\widehat{\Delta}_{i}  & :=\frac{1}{M}\sum_{m=1}^{M}f(X_{\tau_{\ast}%
^{t^{+},g(t,x,t^{+},\ _{m}\xi)}}^{t^{+},g(t,x,t^{+},\ _{m}\xi)})\frac
{\partial}{\partial x^{i}}\left(  \frac{p(t,x,t^{+},g(t,x,t^{+},_{m}\xi
))}{\phi(t,x,t^{+},g(t,x,t^{+},_{m}\xi))}\right)  \label{delprop3}\\
& +\frac{1}{M}\sum_{m=1}^{M}\frac{p(t,x,t^{+},g(t,x,t^{+},_{m}\xi))}%
{\phi(t,x,t^{+},g(t,x,t^{+},_{m}\xi))}\frac{\partial f}{\partial z}%
(X_{\tau_{\ast}^{t^{+},g(t,x,t^{+},\ _{m}\xi)}}^{t^{+},g(t,x,t^{+},\ _{m}\xi)})\cdot\nonumber\\
& \cdot\partial_{y}X_{\tau_{\ast}^{t^{+},g(t,x,t^{+},\ _{m}\xi)}}^{t^{+}%
,y}(g(t,x,t^{+},\ _{m}\xi))\frac{\partial g(t,x,t^{+},\ _{m}\xi)}{\partial
x^{i}},\nonumber
\end{align}
where
$
\frac{\partial}{\partial x^{i}}\frac{p(t,x,t^{+},y)}{\phi(t,x,t^{+},y)}%
,
\frac{\partial g(t,x,t^{+},\ _{m}\xi)}{\partial x^{i}},
\frac{\partial f}{\partial z}%
$ 
can in principle be expressed analytically, and the vector process
$
\partial_{y}X_{s}^{t^{+},y}(\cdot):=\frac{\partial X_{s}^{t^{+},y}}{\partial
y}(\cdot),$ 
$s\geq t^{+},
$ 
can in principle be simulated via a variational system of SDEs (e.g. see
\citet{Pr90}, \citet{MS}, \citet{GG06}). 

In this paper we will prefer the discretized version (\ref{delprop2})  of
(\ref{delprop1}) for our applications.
The algorithm is as follows. We first choose an $h>0$, and sample $_{m}\xi$ for $m=1,\ldots,M$ from the reference (usually normal) density.
Next we simulate for each $m$ a pair of trajectories $_{m}X^\pm,$ which start in
$_mg^\pm$ $:=$ $g(t,x\pm h,t^+,\, _{m}\xi)$ at$\ t^{+},$ and end at the optimal stopping times $_m\tau_\ast^\pm$ $:=$ 
$\tau_{\ast}^{t^{+}, _mg^\pm}.$ 
   Of course the optimal exercise dates $_m\tau_\ast^\pm$  are generally
unknown in practice, but we assume that we have good approximations
$_m\tau^\pm$ at hand, which are constructed via some well known
procedure. For example, in a pre-computation we may construct an exercise
boundary via a regression method (e.g. \citet{LS}), or as
an alternative, we may use the policy iteration method of 
\citet{KS1}, see also \citet{BeSc}. 
As discussed above,  for a particular $\omega$  we have $_m\tau_\ast^+$ $=$ $_m\tau_\ast^-$ provided that $h$ is small enough. For this reason we  take in our simulations simply $_m\tau^-= _m\tau^+,$ where $\tau$ is some approximation of the optimal exercise policy. This pragmatic assumption is justified if the probability 
of the event $_m\tau^-\neq _m\tau^+$ is small enough, i.e. $h$ is small enough.  
For each $m$ we compute also the values $_mp^\pm$ $:=$ $p(t,x\pm \mathfrak{h},t^{+},\,_{m}g^{\pm})$
and $_m\phi^\pm$ $:=$ $\phi(t,x\pm \mathfrak{h},t^{+},\,_{m}g^{\pm})$, and  finally compute the estimate (\ref{delprop2}).

\begin{remark}{\rm
In the previous sections  vector and matrix components are denoted by superscripts,  so  that time parameters of processes  can be denoted by subscripts. In the next sections we depart from this convention and use subscripts for vector and matrix components.  }  
\end{remark}

\section{WKB approximations for transition densities}\label{WKBsec}

\subsection{Recap of WKB theory}
We summarize some results concerning WKB-expansions of parabolic equations (cf. \citet{Ka} for details). 
Let us consider the parabolic  diffusion operator
\begin{equation}
\begin{array}{l}
\frac{\partial u}{\partial  t}+Lu\equiv \frac{\partial u}{\partial t}+\frac{1}{2}\sum_{i,j}a_{ij}\frac{\partial^2 u}{\partial x_i\partial x_j}+
\sum_i b_i\frac{\partial u}{\partial x_i}. \label{PPDE}
\end{array}
\end{equation}
  For simplicity of notation and without loss of generality it is assumed that the diffusion coefficients $a_{ij}
$ and the first order coefficients $b_i$ in (\ref{PPDE})
depend on the spatial variable $x$ only. In the following let $\delta t:=T-t$, and let the functions
\begin{equation*}
(x,y)\rightarrow d(x,y)\ge0,~~(x,y)\rightarrow c_k(x,y),~k\geq 0,
\end{equation*}
be defined on  ${\mathbb R}^n\times {\mathbb R}^n$, with $d^2$ and $c_k,$ $k\ge0$, being smooth. 
Then a set of (simplified) conditions sufficient for pointwise valid WKB-representations 
of the form
\begin{equation}\label{WKBrep}
p(t,x,T,y)=\frac{1}{\sqrt{2\pi \delta t}^n}\exp\left(-\frac{d^2(x,y)}{2\delta t}+\sum_{k= 0}^{\infty}c_k(x,y)\delta t^k\right), 
\end{equation}
for the solution $(t,x)\rightarrow p(t,x,T,y)$ of the final value  problem 
\bea\label{FSu}
\frac{\partial p}{\partial  t}+ Lp&=&0,\quad \mbox{with final value}\\
p(T,x,T,y)&=&\delta(x-y),\quad y\in {\mathbb R}^n\quad \mbox{fixed},\notag 
\eea
is given by
\begin{itemize}
\item[(A)] The operator $L$ is uniformly elliptic in ${\mathbb R}^n$, i.e. as in \eqref{fuc} the matrix norm of $(a_{ij}(x))$ is bounded below and above by $0<\lambda <\Lambda <\infty$ uniformly in~$x$,


\item[(B)] the smooth functions $x\rightarrow a_{ij}(x)$ and $x\rightarrow b_i(x)$ and all their derivatives are bounded.

\end{itemize}
For more subtle (and partially weaker conditions) we refer to \citet{Ka}.
If we add the uniform boundedness condition
\begin{itemize}
\item[(C)] there exists a constant $c$ such that for each multiindex $\alpha$ and for all $1\leq i,j,k\leq n,$ 
\begin{equation}\label{unibd}
{\Big |}\frac{\partial^{\alpha} 
a_{jk}}{\partial x^{\alpha}}(x){\Big |},~{\Big |}\frac{\partial^{\alpha} b_{i}}{\partial x^{\alpha}}{(x)\Big |}\leq c\exp\left(c|x|^2 \right), 
\end{equation}  
\end{itemize}
then the Taylor expansions of the functions $d$ and $c_k$ around $y\in {\mathbb R}^n$ are equal to $d$ and $c_k,k\geq 0$ globally.
I.e.  we have the power series representations
\bea
d^2(x,y)&=&\sum_{\alpha}d_{\alpha}(y)\delta x^{\alpha}\label{dtay}
\\ 
c_k(x,y)&=&\sum_{\alpha}c_{k,\alpha}(y)\delta x^{\alpha},\qquad k\geq 0,\label{ctay}
\eea
where $\delta x:=x-y$.
Note that (C) is implied by the stronger condition that all derivatives in (\ref{unibd}) have a uniform bound.   
Summing up we have the following theorem:
\begin{theorem}
If the hypotheses (A),(B) are satisfied, then the fundamental solution $p$ has the representation
\begin{equation}\label{repwkb}
p(\delta t,x,y)=\frac{1}{\sqrt{2\pi \delta t}^n}\exp\left( -\frac{d^2(x,y)}{2\delta t}+\sum_{k\geq 0}c_k(x,y)\delta t^k\right),
\end{equation}
where $d$ and $c_k$ are smooth functions, which are unique global solutions of the first order differential equations \eqref{ed},\eqref{c01e}, and \eqref{1gaa} below. Especially,
$$(\delta t,x,y)\rightarrow \delta t \ln p(\delta t,x,y) =-\frac{n}{2} \delta t\ln (2\pi \delta t) -\frac{d^2}{2} +\sum_{k\geq 0}c_k(x,y)\delta t^{k+1}$$
is a smooth function which converges to $-\frac{d^2}{2}$ as $\delta t\searrow 0$, where $d$ is the Riemannian distance induced by the line element $ds^2=\sum_{ij}a^{-1}_{ij}dx_idx_j$, where with a slight abuse of notation $(a^{-1}_{ij})$ denotes the matrix inverse of $(a_{ij})$.
If the hypotheses (A),(B) and (C) are satisfied, then in addition the functions $d,~c_k,k\geq 0$ equal their Taylor expansion around $y$ globally, i.e. we have (\ref{dtay})-(\ref{ctay}).
\end{theorem}
The recursion formulas for $d$ and $c_k,~k\geq 0$ are obtained by plugging  the ansatz \eqref{WKBrep} into the parabolic equation (\ref{FSu}), 
and ordering terms with respect to the monoms $\delta t^i=(T-t)^i$ for $i\geq -2$.
By collecting terms of order $\delta t^{-2}$ we obtain
\begin{equation}\label{ed}
d^2=\frac{1}{4}\sum_{ij}d^2_{x_i}a_{ij}d^2_{x_j},
\end{equation}
where $d^2_{x_k}$ denotes the derivative of the function $d^2$ with respect to the variable $x_k$,  with the boundary condition  $d(x,y)=0$ for $x=y.$ 
Collecting terms of order $\delta t^{-1}$ yields
\begin{equation}\label{c01e}
-\frac{n}{2}+\frac{1}{2}Ld^2+\frac{1}{2}\sum_{i} \left( \sum_j\left( a_{ij}(x)+a_{ji}(x)\right) \frac{d^2_{x_j}}{2}\right) \frac{\partial c_{0}}{\partial x_i}(x,y)=0,
\end{equation}
where the boundary condition 
\begin{equation}\label{c01b}
c_0(y,y)=-\frac{1}{2}\ln \sqrt{\mbox{det}\left(a_{ij}(y) \right) }
\end{equation}
determines $c_0$ uniquely for each $y\in {\mathbb R}^n$. Finally,
for $k+1\geq 1$ we obtain
\begin{equation}\label{1gaa}
\begin{array}{ll}
(k+1)c_{k+1}(x,y)+\frac{1}{2}\sum_{ij} a_{ij}(x)\Big(
\frac{d^2_{x_i}}{2}\frac{\partial c_{k+1}}{\partial x_j}
+\frac{d^2_{x_j}}{2} \frac{\partial c_{k+1}}{\partial x_i}\Big)\\
\\
=\frac{1}{2}\sum_{ij}a_{ij}(x)\sum_{l=0}^{k}\frac{\partial c_l}{\partial x_i} \frac{\partial c_{k-l}}{\partial x_j}
+\frac{1}{2}\sum_{ij}a_{ij}(x)\frac{\partial^2 c_k}{\partial x_i\partial x_j}  
+\sum_i b_i(x)\frac{\partial c_{k}}{\partial x_i},
\end{array}
\end{equation}
with boundary conditions
\begin{equation}\label{Rk}
c_{k+1}(x,y)=R_k(y,y) \mbox{ if }~~x=y,
\end{equation}
$R_k$ being the right side of \eqref{1gaa}. For some classical models in finance a global transformation of the diffusion operator to the Laplace operator is possible (at the price of more complicated first order terms however). We observe this in the case of a Libor market model (Section~\ref{appli}). 
The requirement that a transformation $y(x)$ of an operator with second order coefficients $a_{ij}(x)=(\sigma\sigma^\top)_{ij}(x)$ leads to a Laplacian with respect to second order terms in $y$ is equivalent to
\begin{equation}
\sum_{ml}(\sigma\sigma^\top)_{ml}(x)\frac{\partial y_k}{\partial x_l}\frac{\partial y_j}{\partial x_m}=\delta_{jk}  
\end{equation}
where $\delta_{jk}$ denotes the Kronecker delta. If $\sigma$ is invertible it follows directly that the transformation $y(x)$ satisfies the first order matrix equation 
\begin{equation}\label{transsigmainv}
(\frac{\partial y_k}{\partial x_l})=(\sigma^ {-1}_{kl}(x)).
\end{equation}
The latter equation determines the transformation (up to constants, of course) but cannot be integrated in general, and if it can not explicitly in general. However, a necessary and sufficient condition for integrability of \eqref{transsigmainv} in terms of $\sigma$ can be given, where we restrict ourselves to the case of invertible $\sigma$.
\begin{proposition}\label{proptr}
There is a global coordinate transformation for the operator
(\ref{PPDE}) 
such that the second order part of the transformed operator equals the Laplacian, iff $a_{ij}=(\sigma\sigma^\top)_{ij}$ for a (square) matrix function $\sigma$ which satisfies
\begin{equation}
\sum_{l=1}^n\frac{\partial \sigma_{ik}(x)}{\partial x_l}\sigma_{lj}(x)=\sum_{l=1}^n\frac{\partial \sigma_{ij}(x)}{\partial x_l}\sigma_{lk}(x),\quad x\in {\mathbb R}^n. \label{pco}
\end{equation}
\end{proposition}
The latter fact is also observed and proved in \citet{AS}.
If the condition of Proposition~\ref{proptr} is satisfied, then coordinate transformation leads to second order coefficients of the form $a_{ij}\equiv \delta_{ij},$ so that the solution of (\ref{ed}) becomes
\begin{equation*}
d^2(x,y)=\sum_i (x_i-y_i)^2.
\end{equation*} 
If conditions (A), (B), (C), and (\ref{pco}) hold, then in the transformed coordinates, explicit formulas for the coefficient functions $c_k, k\geq 0$ can be computed via the formulas
\begin{eqnarray}\label{solc0ck}
c_0(x,y)&=&\sum_i(y_i-x_i)\int_0^1 b_i(y+s(x-y))ds,\nonumber\\
c_{k+1}(x,y)&=&\int_0^1 R_k(y+s(x-y),y)s^{k}ds,
\end{eqnarray}
with $R_k$ being  the right-hand-side of \eqref{1gaa} where $a_{ij}$ $=$ $\delta_{ij}.$ Similar formulas are obtained in \citet{AS}.
In \citet{Ka} it is shown in addition how the coefficients $c_k$ can be computed explicitly in terms of power series approximations of the diffusion coefficients $a_{ij}$ and $b_i$. However, in high dimensional models such as the Libor market model direct computation of the coefficients $c_k$ seems more feasible as it turns out that the computation up to the coefficient $c_1$ is sufficient for our purposes.

\subsection{Error estimates}\label{errors}

We now study the approximation error of a truncated WKB expansion (and its derivatives), which is essential for convergence of the Monte Carlo schemes. In this respect we will show how the derivatives (up to second order) of the product value function with respect to the underlyings computed by means of a truncated WKB-expansion converge in supremum norm and H\"older norms. 
Let us consider a WKB-approximation of the fundamental solution $p$ of the form
\begin{equation}\label{approxpk} 
p_l( t,x,T,y)=\frac{1}{\sqrt{2\pi \delta t}^n}\exp\left(-\frac{d^2(x,y)}{2\delta t}+\sum_{k= 0}^{l}c_k(x,y)\delta t^k\right), 
\end{equation}
i.e. we assume that the coefficients $d^2$ and $c_k,~0\leq k\leq l$ have been computed up to order $l$ (recall that $\delta t=T-t$ for the sake of brevity). 
Let  us denote the domain of the Cauchy problem by $D=(0,T)\times {\mathbb R}^n$. For integers $n\geq 0$ and real numbers $\delta \in (0,1)$ let $C^{m+\delta/2,n+\delta}(D)$ be the space of $m$ ($n$) times differentiable functions such that the $m$th ($n$th) derivative with respect to time (space) is H\"older continuous with exponent $\frac{\delta}{2}$ ($\delta$). Furthermore,  $|.|_{m+\delta/2,n+\delta}$ denote the natural norms associated with these function spaces. It is well-known that in case of our assumptions (A) and (B) the fundamental solution $p$  satisfies the a priori estimate
\begin{equation}\label{aprioris}
|p(t,x,T,y)|\leq C(T-t)^{-n/2}\exp\left(-\frac{\lambda_0 |x-y|^2 }{2(T-t)} \right), 
\end{equation} 
for some generic constant $C$ and some $\lambda_0$ which is less or equal than the lower bound $\lambda$ in assumption (A) above. We call a WKB-approximation $p_l$ of the fundamental solution $p$ admissible, if it satisfies the a priori estimate \eqref{aprioris}. The WKB-approximation $p_0$ is always admissible while the proof of theorem 9 (cf. \cite{Ka}) that $p_l$ is admissible if $l\geq l_0$ where $l_0$ is some natural number depending on the coefficient functions and can be computed by comparison of the WKB-expansion and the Levy-expansion. For lower $l$ admissibility has to be ensured for each model. In the Libor market model admissibility for $l=1$ is ensured.  
As a consequence of Safanov's theorem (cf. \citet{Kr}) we have
\begin{theorem}
Assume that $(A), (B), \mbox{ and } (C)$ are satisfied and let $h\in C^{2+\delta}\left({\mathbb R}^n\right)$ and $f\in C^{\delta/2,\delta}(D)$. If 
\begin{equation}\label{potential}
 c\leq -\lambda~\mbox{ for some }~\lambda >0,
\end{equation}
then the Cauchy problem
\begin{equation}\label{cauchydrift}
\left\lbrace \begin{array}{ll} \frac{\partial w}{\partial  t}+\frac{1}{2}\sum_{ij}a_{ij}(x)\frac{\partial^2 w}{\partial x_i\partial x_j}+\sum_i b_i(x)\frac{\partial w}{\partial x_i}+c(x)w=f(t,x) \mbox{ in } D\\
\\
w(T,x)=h(x) \mbox{ for } x\in {\mathbb R}^n
\end{array}\right.
\end{equation}
has a unique solution $w$, and there exists a constant $c$ depending only on $\delta$, $n$ $\lambda,\Lambda$ and $K=\max\{|a|_{\delta},|b|_{\delta},|c|_{\delta}\}$ such that
\begin{equation}\label{aprioripara}
|w|_{1+\delta/2,2+\delta}\leq c\left[|f|_{\delta/2,\delta}+|h|_{2+\delta}\right]. 
\end{equation}
\end{theorem}
In order to analyze the truncation error of the Cauchy problem with data $h$ we consider the function
\begin{equation*}
u^{\Delta}(t,x)= u( t,x)\\
-u_l( t,x),\quad\mbox{where}
\end{equation*}
\begin{equation}\label{ul}
u( t,x)=\int_{{\mathbb R}^n}h(y)p(t,x,T,y)
\quad\mbox{and}\quad
u_l( t,x)=\int_{{\mathbb R}^n}h(y)p_l( t,x,T,y)dy.
\end{equation}
We say that $u_l$ in \eqref{ul} is  admissible if $p_l$ is admissible. It is now possible to derive different error estimates in strong norms depending on which Greeks we want to control on which level of regularity. 
\begin{theorem} Assume that conditions (A), (B), and (C) hold and that $h\in C^{2+\delta}({\mathbb R}^n)$ and assume that $u_l$ is admissible. Then 
\begin{equation*}
|u( t,x)-u_l( t,x)|_{1+\delta/2,2+\delta}\in O( t^{l-\frac{\delta}{2}}).
\end{equation*}
\end{theorem}

\begin{proof}
Let $w(t,x)=e^{-r t}u^{\Delta}(t,x)$ with $r$ constant and $w_l(t,x)=e^{-r t}u_l(t,x)$. Since 
\begin{equation*}
\begin{array}{ll}
\frac{\partial u^{\Delta}}{\partial  t}+Lu^{\Delta}=
-\frac{\partial u_l}{\partial  t}-Lu_l=:f_{u_l}(t,x), \mbox{ we have}\\
\\
\frac{\partial w}{\partial  t}+Lw+rw=\frac{\partial w_l}{\partial  t}+Lw_l+rw_l=
e^ {rt}\left( -\frac{\partial u_l}{\partial  t}-Lu_l\right) =:e^ {rt}f_{u_l}(t,x).
\end{array}
\end{equation*}
Admissibility of $u_l$ and an argument similar to that of Krylov ensures that the right side of the latter equation can be measured in the norm $|.|_{\delta/2,\delta}$.
Hence  we can apply the estimate \eqref{aprioripara} to the function
$
w(t,x)=e^{-r t}u^{\Delta}(t,x)
$ 
for a constant $r>0$ and we get (after dividing by $e^ {rt}$)
 \begin{equation}\label{err1}
|u^{\Delta}|_{1+\delta/2,2+\delta}
\leq c|f_{u_l}|_{\delta/2,\delta}.
\end{equation}
In order to compute the term on the right side of \eqref{err1} we can plug \eqref{approxpk} into the left-hand side of (\ref{FSu}) the parabolic equation satisfied by the exact fundamental solution $p$. However in order to see how the higher order terms behave exactly we plug in
\begin{equation}
p( t,x,y)=\frac{1}{\sqrt{2\pi t}^n}\exp\left(-\frac{d^ 2(x,y)}{2 t}+\sum_{k=0}^ {l}c_k(x,y) t^k +R_{l+1}(t,x,y) \right),
\end{equation}
where
\begin{equation}
R_{l+1}(t,x,y)=\sum_{k=l+1}^ {\infty}c_k(x,y) t^k= O( t^{l+1}).
\end{equation}
We get
\begin{equation*}\label{erreq}
\begin{array}{ll}
\frac{\partial p}{\partial t}+\frac{1}{2}\sum_{ij}a_{ij}\frac{\partial^2 p_l}{\partial x_i \partial x_j}+\sum_i b_i\frac{\partial p}{\partial x_i}=\\
\\
= t^ l{\Bigg (}(l+1)c_{l+1}+\frac{\partial}{\partial t}R_{l+1}+\frac{1}{2}\sum_{ij}a_{ij}{\Big(}-\frac{d^2_{x_i}}{2}\left( \frac{\partial}{\partial x_j}\frac{R_{l+1}}{ t}\right) -\frac{d^2_{x_j}}{2}\left( \frac{\partial}{\partial x_i}\frac{R_{l+1}}{ t}\right) \\
+(c_{l,x_i}+R_{l+1,x_i})(c_{l,x_j}+R_{l+1,x_j}){\Big )} +L(c_l+R_l){\Bigg )}p= O(t^ l)p,\quad t\downarrow 0.
\end{array}
\end{equation*}
 Applying a priori estimates for $p$ we get the result. 
\end{proof}
\begin{remark}{\rm
 A more intricate analysis shows that for practical purposes it is possible to remove the admissibility condition above, if we approximate the Cauchy problem by a Dirichlet problem with a large but spatially bounded domain (a natural step from the numerical point of view). However, since this involves an additional analysis of an integral equation corresponding to the boundary condition we go not into the details here. Generalizations to estimates which include Taylor-expansions of the WKB-coefficients will also be considered elsewhere. 
The assumption $h\in C^{2+\delta}({\mathbb R}^n)$ can be weakened to H\"older continuous pay-offs if we abstain from controlling the $\Theta$ Greek (sensitivity with respect to time) up to maturity. 
The case where $c_k$ are computed up to $k=1$ is the first case where the truncation error for first and second derivatives converges to zero (in supremum norm with order $O(\delta t)$ and in H\"older- extension of supremum norm with order $O(\delta t)^{1-\frac{\delta}{2}}$). This implies that our Monte Carlo computation scheme for the Greeks converges. }
\end{remark}

\begin{remark}{\rm
We can easily see how the boundedness of the constant $M_5$ in Theorem~\ref{The} is controled in situations where $p$ is 
WKB-approximated 
by $p^l,$ $l\ge0,$ and where the  prior is chosen according to (\ref{zeta})
in Section~4. For simplicity we here assume that the problem is reduced to the form $a_{ij}=\delta_{ij}.$ Then the logarithmic derivative  of the WKB-expansion (cf. (35)) of the density $p$ 
 takes the form
\begin{eqnarray*}
\frac{1}{p}\frac{\partial p}{\partial x} &=&-\frac{x-y}{t}+\sum_{k\ge0}\frac{\partial }{\partial x}c_k(x,y)t^k
=\frac{1}{p^l}\frac{\partial p^l}{\partial x} +O(t^{l+1}).
\end{eqnarray*}
The logarithmic derivative of the lognormal prior (\ref{zeta}) takes the form
\begin{equation}\label{phii}
\frac{1}{\phi}\frac{\partial \phi}{\partial x} =-\frac{x-y}{t}+\frac{\partial }{\partial x} c^{\phi}_0(x,y), \quad {\rm hence} 
\end{equation}
\begin{equation*}
\frac{1}{p^l}\frac{\partial p^l}{\partial x}-\frac{1}{\phi}\frac{\partial \phi}{\partial x}=\frac{\partial}{\partial x}(c_0(x,y)-c^{\phi}_0(x,y))+O(t),
\end{equation*}
i.e., in the difference  the first terms  cancel out. So for small time $t$ the main contibution to the constant $M_5$ is the difference of $\frac{\partial}{\partial x}(c_0-c^{\phi}_0)$ which does not depend on $t$.  Note that if we freeze \eqref{phii} at $x_0$ say, which leads 
to the ´´naive´´ estimator (\ref{Naiv}), the difference contains a term of order of $O(t^{-1})!$
}
\end{remark}

\section{Applications to the Libor market model}\label{appli}

We consider a Libor market model with respect to a tenor structure $0<T_1\ldots <T_{n+1}$ in the terminal measure $P_{{n+1}}$ 
(induced by the terminal zero coupon bond $B_{n+1}(t)$). The dynamics of the forward Libors $L_{i}(t)$,
defined in the interval $[{0},T_{i}]$ for $1\leq i\leq n,$ are governed
by the following system of SDE's (e.g., see \citet{Jam97}),
\begin{equation}
dL_{i}=-\sum_{j=i+1}^{n}\frac{\delta _{j}L_{i}L_{j}\,\gamma
_{i}^\top\gamma _{j}}{1+\delta _{j}L_{j}}\,dt+L_{i}\,\gamma_{i}^\top dW^{(n+1)}=:
\mu_i(t,L)L_i+L_{i}\,\gamma_{i}^\top dW^{(n+1)}, \label{LMM}
\end{equation}
where $\delta _{i}=T_{i+1}-T_{i}$  are  day count fractions and
$\displaystyle t\rightarrow\gamma _{i}(t)=(\gamma
_{i,1}(t),\ldots,\gamma _{i,d}(t))$, $0$ $\le$  $t$ $\le$ $T_i,$
are  bounded and smooth enough deterministic volatility
vector functions. 
We denote the matrix with rows $\gamma_i^\top$ by $\Gamma$ and assume that
$\Gamma$ is invertible. In what follows we assume that $\Gamma(t)\equiv\Gamma$
does not depend on $t$. The case of time-dependent volatility is discussed
in \citet{KKS}.
In (\ref{LMM}), $(W^{(n+1)}(t)\mid 0\leq $ $t\leq T_{n})$ is a
standard $d$-dimensional Wiener process under the measure $P_{n+1}$ 
with $d,$ $1\leq d\leq n,$ being the number of driving factors.
In what follows we consider the full-factor Libor model with $d=n$
in the time interval $[0,\,T_1)$.

\subsection{WKB approximations for the Libor kernel}\label{Libor-WKB}

Let us transform the dynamics of (\ref{LMM}) to $K_i:=\ln L_i$, $1\leq i\leq n$,
\begin{equation}
dK_i=\frac{1}{L_i}dL_i-\frac{1}{2L_i^2}d\langle L_i\rangle=
\left(-\frac{\gamma_i^\top\gamma_i}{2}+\mu_i(t,e^{K_1},\ldots, e^{K_n})\right)dt
+\gamma_i^\top dW^{(n+1)}, \label{KWKB}
\end{equation}
where $\mu$ is given in (\ref{LMM}).
Note that the coefficients of the generator corresponding to process $K$ are bounded and satisfy the WKB assumptions (A) (B) in Section~5.1. Hence, we may apply a WKB approximation to the transition density of the process (\ref{KWKB}). 

By the transformation $Y:=\Gamma^{-1}K$ we obtain the process
\begin{equation}\label{YWKB}
dY_i=\mu^Y_i(t,Y)dt+dW_i^{(n+1)},\quad 1\leq i\leq n,\quad {\mbox{\rm where}}
\end{equation}
$$
\mu_i^Y(t,Y)=V_i+\sum_{j=1}^{n}\Gamma_{ij}^{-1}\mu_j(t,e^{(\Gamma Y)_1},\ldots, e^{(\Gamma Y)_n}),
\quad
V_i=-\sum_{j=1}^n\Gamma_{ij}^{-1}\frac{\vert\gamma_j\vert^2}{2},
$$
for which the generator has a  Laplacian diffusion term, which leads to technically more convenient
expressions in the respective WKB expansion.

The  situation of time independent $\gamma$ (hence bounded $\mu^Y$) is exemplified in case study Section~\ref{num}, where 
the transition density $p^Y$ is approximated  and subsequently transformed to an approximated transition density $p^L$ of the Libor process.  Below we spell out the ingredients for computing the corresponding WKB coefficients 
 $c_0$ and $c_1$ according to (\ref{solc0ck}) to be exploited in Section~\ref{num}.  
Using the notations
$$
F_l(s,x,y):=\frac{1}{(\Gamma(x-y))_l}
\ln\frac{1+\delta_l e^{(\Gamma x)_l}}{1+\delta_l e^{(\Gamma y)_l}},\quad 1\leq l\leq n,
$$
and $a:=(\gamma_i^\top\gamma_j)_{i,j=1}^n,$ we may write,
\begin{equation}
c_0(s,x,y)=\sum_{i=1}^{n}V_i(y_i-x_i)
+\sum_{i=1}^{n}\sum_{j=1}^{n}\Gamma_{ij}^{-1}(y_i-x_i)\sum_{l=j+1}^{n}
a_{jl}F_l(s,x,y),
\end{equation}
\begin{eqnarray*}
\frac{\partial c_0}{\partial x_p}(s,x,y)&=&
-V_p+\sum_{j=1}^{n}\Gamma_{pj}^{-1}\sum_{l=j+1}^na_{jl}F_l(s,x,y)-\\
&&\sum_{i=1}^{n}\sum_{j=1}^{n}\Gamma_{ij}^{-1}(y_i-x_i)\sum_{l=j+1}^{n}a_{jl}
\frac{\partial F_l(s,x,y)}{\partial x_p},
\end{eqnarray*}
\begin{eqnarray*}
\frac{\partial^2 c_0}{\partial x_p^2}(s,x,y)&=&
2\sum_{j=1}^{n}\Gamma_{pj}^{-1}\sum_{l=j+1}^{n}a_{jl}\frac{\partial F_l(s,x,y)}{\partial x_p}-\\
&&\sum_{i=1}^{n}\sum_{j=1}^{n}\Gamma_{ij}^{-1}(y_i-x_i)\sum_{l=j+1}^{n}a_{jl}
\frac{\partial^2 F_l(s,x,y)}{\partial x_p^2},
\end{eqnarray*}
where
$$
\frac{\partial F_l(s,x,y)}{\partial x_p}=
\frac{\Gamma_{lp}(s)}{(\Gamma(x-y))_l}
\left(
\frac{\delta_l e^{(\Gamma x)_l}}{1+\delta_l e^{(\Gamma x)_l}}-F_l(s,x,y),
\right)\quad\mbox{and}
$$

\begin{eqnarray*}
\frac{\partial^2 F_l(s,x,y)}{\partial x_p^2}&=&
\frac{2\Gamma_{lp}^2}{(\Gamma(x-y))_l^2}
\left(
F_l(s,x,y)-\frac{\delta_l e^{(\Gamma x)_l}}{1+\delta_l e^{(\Gamma x)_l}}
\right)\\
&&+
\frac{\Gamma_{lp}^2}{(\Gamma(x-y))_l}
\frac{\delta_le^{(\Gamma x)_l}}{(1+\delta_le^{(\Gamma x)_l})^2}.
\end{eqnarray*}

We finally obtain $p^L(s,u,t,v)$ by density transformation formula,
\begin{equation*}
p^L(s,u,t,v)=p^Y(s,S_s^{-1}(u),t,S_t^{-1}(v))\left\vert\frac{\partial S_t^{-1}(v)}{\partial v}\right\vert
\end{equation*}
with
$$
S_t^{-1}(v):=\Gamma^{-1}(t)(\ln v_1,\dots,\ln v_n)^\top.
$$
For simplicity  in the case study below 
we assume that the matrix $\Gamma$ is upper triangular and does not depend on $t$.
We then  have, 
\begin{eqnarray*}
p^L(s,u,t,v)&=&\frac{1}{\sqrt{2\pi(t-s)}^{n}}\prod_{i=1}^n\frac{\Gamma_{ii}^{-1}}{v_i}
\exp\left(
-\frac{\left(\Gamma^{-1}(\ln \frac{v_1}{u_1},\dots,\ln \frac{v_n}{u_n})^\top\right)^2}{2(t-s)}\right.\\
&&
\left.+\sum\limits_{k=0}^{\infty}c_k(s,S^{-1}(u),S^{-1}(v))(t-s)^k
\right).
\end{eqnarray*}

\subsection{Case study}\label{num}

We now illustrate the estimators (\ref{E1}) and (\ref{nonexpl}) in Section \ref{Estimators} and 
the estimators (\ref{E11}) and (\ref{delprop2}) in Section \ref{Bermudan} by computing European and Bermudan
swaptions and Deltas in a Libor market model.
A (payer) swaption contract with maturity ${T}_i$ and strike
$\theta$ with principal $\$1$ gives the right to contract at
${T}_i$ for paying a fixed coupon $\theta$ and receiving floating
Libor at the settlement dates ${T}_{i+1}$,$\dots$,${ T}_n$. The discounted payoff of the contract is thus given by
\begin{equation}\label{swap-u}
f_i(L(T_i))=
\frac{1}{B_{n+1}(T_i)}\sum\limits_{j=i}^{n}B_{j+1}(T_i)\left(\delta_j
L_j(T_i)
-\theta\right)^+.
\end{equation}
For our experiments we take in (\ref{LMM}),  $\delta_i\equiv0.5,$ 
$L(0)$ $=$ $3.5\%$ flat, and constant volatility loadings,
$
\gamma_{i}(t)\equiv 0.2e_i,
$
where $e_{i}$ are $n$-dimensional unit vectors decomposing an
input correlation matrix $\rho$,
\begin{align}
\rho_{ij}=\exp\Big[\frac{|j-i|}{n-1}\ln\rho_{\infty}\Big], \quad
1\leq i,j\leq n,
\label{ExCor2aa}%
\end{align}
with 
$\rho_{\infty}=0.3$
(for more
general correlation structures we refer to \citet{SchoeBook}).
We consider at-the-mondey European swaptions with maturity $T_1$ and 
at-the-money Bermudan swaptions with 10 annual exercise possibilities, 
starting from $T_1$, hence $\theta=3.5\%$ in (\ref{swap-u}). 

For the Bermudan swaptions a good approximation of the optimal stopping policy is constructed by 
Andersen's method (strategy II, see \citet{An1}),
\begin{eqnarray}\label{tau-A}
\tau^{T_1,L(0)}_A:=\inf\left\{
i:\;T_i\geq T_1,\,T_i\in {\mathcal T},\,B_{n+1}(T_i)f_i(L_{T_i}^{T_1,L(0)})\geq H_i+ \right.&&\nonumber\\
\left.
\max_{j\geq i,T_j\in\ {\mathcal T}}E^i\left[B_{n+1}(T_j)f_j(L_{T_j}^{T_1,L(0)})\right]
\right\},&&
\end{eqnarray}
where
$L^{0,L(0);\,j}_{T_i}:=L_j(T_i)$
in line with Sections \ref{Basic} and  \ref{Bermudan}, and  ${\mathcal T}$ $:=$ $\{T_1,$ $T_3,$ $T_5,$ $\ldots,$ $T_{19}\}$ is the 
set of possible exercise dates.
The conditional expectations in (\ref{tau-A})
can be computed accurately in  closed-form (see, e.g., \citet{SchoeBook}).  
Further 
in (\ref{tau-A}), $H$ is a constant vector computed by backward optimization over a set of pre-simulated trajectories, as proposed by 
\citet{An1}. 
In Table~3, column~2, we display the Bermudan prices $\widehat{u}_{\mbox{\tiny ex}}^{\mbox{\tiny low}}$ due to stopping strategy $\tau_A.$
Upper estimations $\widehat{u}_{\mbox{\tiny ex}}^{\mbox{\tiny up}}$ are constructed from $\widehat{u}_{\mbox{\tiny ex}}^{\mbox{\tiny low}}$ by the
dual approach, developed in \citet{R} and \citet{HK}, see Table~3, column~1. 
As we see, the distance between lower and upper Bermudan estimates does not exceed 0.5\% (relative to the values).

The Libor transition kernel $p^L(s,x,t,y)$ shows to have  a pronounced "delta-shaped" form. 
%
Because of this,
it is very important for efficiency of the estimators in Sections \ref{Estimators}-\ref{Bermudan} to find a suitable proxy  density $\phi.$ 
We take for $\phi$ 
the transition kernel of a lognormal approximation $L^{\mbox{\tiny \rm lgn} }$, obtained from the Libor process (\ref{LMM}) by freezing the coefficients at the initial time $s$,
\begin{equation}\label{Lg}
L_t^{^{\mbox{\tiny \rm lgn}}s,x;\,i}(\xi)=x_i\exp(\xi_i),
\end{equation}
where $\xi$ is a $n$-dimensional Gaussian vector with
\begin{eqnarray}\label{num-xi}
&&E\xi_i=(t-s)\left(\frac{\vert\gamma_i\vert^2}{2}
-\sum_{j=i+1}^n\frac{\vert\gamma_i\vert\vert\gamma_j\vert\rho_{ij}\delta_j x_j}{1+\delta_j x_j}\right)=:\mu_i^{\mbox{\tiny \rm lgn}}(s,t,x),\nonumber\\
&&Cov(\xi_i,\xi_j)=\Gamma_{ij},\quad 1\leq i,j\leq n.
\end{eqnarray}
The transition density of $L^{\mbox{\tiny \rm lgn} }$ is then given by
\begin{eqnarray*}
\phi(s,u,t,v)
&:=&\frac{1}{\sqrt{2\pi(t-s)}^n}\prod_{i=1}^n\frac{\Gamma_{ii}^{-1}}{v_i}\times\\
&&\exp\left(-\frac{\left|\Gamma^{-1}((\ln \frac{v_1}{u_1}\dots\ln \frac{v_n}{u_n})-\mu^{\mbox{\tiny \rm lgn}}(s,t,x))^T\right|^2}{2(t-s)}\right),
\end{eqnarray*}
with $\mid \cdot \mid$ denoting the Euclidean norm.
So, in order to sample from  density $\phi$, we simulate 
via (\ref{Lg})-(\ref{num-xi}) the lognormal samples 
$$
_m\zeta=L_{T_1}^{^{\mbox{\tiny \rm lgn}} 0,L(0)}(_m\xi)=:g(0,L(0),T_1,_m\xi),\quad m=1,\ldots,M.
$$
As a (more accurately approximated)   Libor transition kernel, we use WKB approximation 
$p^L_0$ and $p^L_1$.
We endow the corresponding estimators with superscripts $_0$ and $_1$ respectively. 

European and Bermudan prices and Deltas via the estimators (\ref{E1}), (\ref{nonexpl}), (\ref{E11}), (\ref{delprop2})
are given in Tables 1-4. These results are compared with corresponding estimates due to "exact" Libor trajectories, simulated by a log-Euler scheme with small time step $\Delta t$ (we take
$\Delta t=\delta_i/5$ for Europeans, $\Delta t=\delta_i/10$ for Bermudans). The "exact" estimates are endowed with the superscript {\tiny ex}. For comparison, the  corresponding estimates due to the standard lognormal Libor approximation $L^{\mbox{\tiny \rm lgn}}$ are computed as well.
In order to keep standard deviations within 0.5\% relative (to the values) we take $h=3.5\times 10^{-5}$, $M$ $=$ $5\times 10^5.$
As we see,
the WKB approximation with only two coefficients, $c_0$ and $c_1$, 
provides a very close estimate of the European swaptions and Deltas, also for 
large maturities. The distance between the values simulated via "exact" Libor trajectories 
and the corresponding values due to the  WKB approximation
is smaller than 0.5\% relative to the value. In contrast, the lognormal estimators 
$\widehat{I}_{\mbox{\tiny \rm lgn}}$, $\frac{\partial \widehat{I}_{\mbox{\tiny \rm lgn}}}{\partial x_i}$,
$\widehat{u}_{\mbox{\tiny \rm lgn}}$ and $\frac{\partial \widehat{u}_{\mbox{\tiny \rm lgn}}}{\partial x_i}$
give  an acceptable approximation only for $T_1\leq 2$.\\
\ \\
{\bf Table 1.} European swaptions (values in basis points)\\
\begin{tabular}{| c || c | c | c | c |}
\hline
$T_1$ &
$\widehat{I}_{\mbox{\tiny ex}}$ (SD) &
$\widehat{I}_{\mbox{\tiny \rm lgn}}$ (SD) &
$\widehat{I}_{0}$ (SD) &
$\widehat{I}_{1}$ (SD) 
\\
\hline
\hline
1.0 &
178.9(0.4) &
179.0(0.4) &
181.6(0.4) &
178.9(0.4)
\\
2.0 &
245.3(0.6) &
246.5(0.6) &
251.4(0.6) &
244.3(0.6)
\\
5.0 &
351.3(1.0) &
359.8(1.0) &
376.4(1.1) &
352.7(1.0)
\\
10.0 &
429.6(1.5) &
451.4(1.6) &
495.6(1.7) &
431.8(1.4)
\\
\hline
\end{tabular}
\\
\ \\
\ \\
{\bf Table 2.} European Deltas (values in basis points)\\
\begin{tabular}{| c || c | c | c | c |}
\hline
$T_1$ &
$\widehat{\frac{\partial I_{\mbox{\tiny ex}}}{\partial x_n}}^{(h)}$ (SD) &
$\widehat{\frac{\partial I_{\mbox{\tiny \rm lgn}}}{\partial x_n}}^{(h)}$ (SD) &
$\widehat{\frac{\partial I_0}{\partial x_n}}^{(h)}$ (SD) &
$\widehat{\frac{\partial I_1}{\partial x_n}}^{(h)}$ (SD) 
\\
\hline
\hline
1.0 &
1768.3(2.8) &
1774.2(2.8) &
1794.9(2.9) &
1770.7(2.8)
\\
2.0 &
1726.4(2.9) &
1732.1(2.9) &
1729.0(2.9) &
1729.0(2.9)
\\
5.0 &
1599.6(3.2) &
1615.9(3.3) &
1722.5(3.5) &
1597.0(3.2)
\\
10.0 &
1417.1(3.8) &
1474.0(4.2) &
1668.0(4.7) &
1422.7(3.9)
\\
\hline
\end{tabular}
\\
\ \\
\ \\
{\bf Table 3.} Bermudan swaptions (values in basis points)\\
\begin{tabular}{| c || c | c || c | c | c |}
\hline
$T_1$ &
$\widehat{u}^{\mbox{\tiny low}}_{\mbox{\tiny ex}}$ (SD) &
$\widehat{u}^{\mbox{\tiny up}}_{\mbox{\tiny ex}}$ (SD) &
$\widehat{u}_{\mbox{\tiny \rm lgn}}$ (SD) &
$\widehat{u}_{0}$ (SD) &
$\widehat{u}_{1}$ (SD) 
\\
\hline
\hline
1.0 &
351.2(0.7) &
352.5(1.0) &
350.9(0.7) &
354.7(0.7) &
351.2(0.7)
\\
2.0 &
388.4(0.8) &
389.8(1.0) &
388.2(0.8) &
396.6(0.8) &
387.3(0.8)
\\
5.0 &
461.5(1.1) &
463.4(1.3) &
466.3(1.1) &
492.9(1.1) &
460.8(1.1)
\\
10.0 &
523.7(1.6) &
524.8(1.7) &
543.6(1.7) &
601.2(1.7) &
523.6(1.5)
\\
\hline
\end{tabular}
\\
\ \\
\ \\
{\bf Table 4.} Bermudan Deltas (values in basis points)\\
\begin{tabular}{| c || c | c | c | c |}
\hline
$T_1$ &
$\widehat{\frac{\partial u_{\mbox{\tiny ex}}}{\partial x_n}}^{(h)}$ (SD) &
$\widehat{\frac{\partial u_{\mbox{\tiny \rm lgn}}}{\partial x_n}}^{(h)}$ (SD) &
$\widehat{\frac{\partial u_0}{\partial x_n}}^{(h)}$ (SD) &
$\widehat{\frac{\partial u_1}{\partial x_n}}^{(h)}$ (SD) 
\\
\hline
\hline
1.0 &
2709.2(3.5) &
2720.9(3.5) &
2747.2(3.5) &
2709.2(3.5)
\\
2.0 &
2631.1(3.5) &
2630.5(3.5) &
2700.7(3.6) &
2628.6(3.5)
\\
5.0 &
2392.9(3.7) &
2407.7(3.8) &
2561.9(4.0) &
2398.0(3.8)
\\
10.0 &
2101.5(4.4) &
2152.5(4.7) &
2443.4(5.3) &
2111.5(4.4)
\\
\hline
\end{tabular}

\begin{remark}{\rm
The values in Tables~1--4 are computed using a second order Taylor approximation of $c_1(x,y)$ around $x,$ where  
$c_1(x,x)$, the derivatives $
\frac{\partial c_1}{\partial y_i}(x,x)
$  and
$
\frac{\partial^2 c_1}{\partial y_i\partial y_j}(x,x),
$ are computed (using finite differences) prior to the Monte Carlo simulation.
}
\end{remark}

\subsubsection*{Computational time}

\noindent By using the  new estimators we avoid step-by-step Euler simulation of the Libor process in the time interval $[0,T_1].$ Generally, the cost of Euler stepping 
up to  $T_1$ is proportional to $T_1/\Delta t,$ whereas the cost of the ''direct estimators'' (\ref{nonexpl}) and (\ref{delprop2}) is independent of $T_1.$ In particular, in the present Libor case, Euler stepping up to $T_1$ requires a cost proportional to $n^2\frac{T_1}{\Delta t}$ times  the cost of computing the (possibly virtual) pay-off at $T_1.$ In comparison, the cost of  simulating  estimators   (\ref{nonexpl}) and (\ref{delprop2}) is 
proportional to $n^2$ times  the cost of the pay-off at $T_1.$

In Figure 2 we compare for different $T_1$ the CPU time (per sample) needed for computing the values in Tables~2,4  using WKB based estimators (\ref{nonexpl}) and (\ref{delprop2}) with  the CPU time
required for computing the estimates via straightforward Euler stepping of "exact" Libor trajectories
up to $T_1.$  
We conclude that, particularly for larger $T_1,$  the efficiency gain is quite high in the European case, and  still considerable in   the Bermudan case.

\begin{remark}{\rm


\citet{FrKa} and  \citet{FJ} propose simulation schemes which 
improve upon Euler SDE simulation and  allow for taking larger time steps for obtaining the  same accuracy. 
  Assuming  that such a scheme requires a time step of order, say  $O(\sqrt \Delta t),$ instead of  $O(\Delta t)$ 
for the same accuracy, it is clear that, for example in the European case, the gain of our method  with respect to this one 
is still order of   $O(T/{\sqrt\Delta t}).$ 
}
\end{remark}

%
%
\begin{figure}
\vspace*{0cm} \caption{CPU time (seconds) for simulating European (left) and Bermudan (right) Deltas for different $T_1$
by log-Euler Libor simulation (solid line) and by  WKB density approximation (with $c_0$ and ${c}_1$) (dash line).}
\end{figure}


\end{document}